\documentclass[journal]{IEEEtran}
\usepackage{amsmath,amscd,graphics,latexsym,multicol,cite,ulem,epsf}
\usepackage{verbatim}
\usepackage{soul}
\usepackage{graphicx, subfigure}
\usepackage{amssymb}
\usepackage{amsthm}
\usepackage{mathabx}
\usepackage{dsfont}
\usepackage{bbm}
\usepackage[utf8]{inputenc}
\usepackage{caption}

\newtheorem{thm}{Theorem}
\usepackage{accents}
\newcommand{\qeq}{\accentset{(a)}{=}}

\begin{document}
\title{Tier Association Probability and Spectrum Partitioning for Maximum Rate Coverage in Multi-tier Heterogeneous Networks}
\author{Sanam Sadr, \textit{Student Member, IEEE}, and Raviraj S. Adve, \textit{Senior Member, IEEE}\thanks{The authors are with the Edward S. Rogers Sr. Department of Electrical and Computer Engineering, University of Toronto, ON, Canada (email: \{ssadr, rsadve\}@comm.utoronto.ca). They would like to acknowledge the financial support of the National Science and Engineering Research Council (NSERC) of Canada and TELUS.}}
\date{}
\maketitle
\begin{abstract}
For a wireless multi-tier heterogeneous network with orthogonal spectrum allocation across tiers, we optimize the association probability and the fraction of spectrum allocated to each tier so as to maximize rate coverage. In practice, the association probability can be controlled using a biased received signal power. The optimization problem is non-convex and we are forced to explore locally optimal solutions. We make two contributions in this paper: first, we show that there exists a relation between the first derivatives of the objective function with respect to each of the optimization variables. This can be used to simplify numerical solutions to the optimization problem. Second, we explore the optimality of the intuitive solution that the fraction of spectrum allocated to each tier should be equal to the tier association probability. We show that, in this case, a closed-form solution exists. Importantly, our numerical results show that there is essentially zero performance loss. The results also illustrate the significant gains possible by jointly optimizing the user association and the resource allocation.
\end{abstract}

\section{Introduction}
Multi-tier heterogeneous networks are receiving strong consideration as the means to meet the huge anticipated growth in traffic demand due to mobile voice, video and wireless data. Such networks comprise multiple tiers of access points (APs) where each tier differs in spatial density and transmit power. Importantly, it is expected that these APs will be deployed in a non-deterministic manner.

The analysis and design of heterogeneous networks necessitate a tractable model for the random AP locations~\cite{ghosh:12}. One common 2-D spatial model that captures this randomness is the Poisson point process (PPP) characterized by only one parameter, $\lambda$, the density of the nodes. In this model, node locations are independent and the number of nodes in disjoint regions are independent random variables. Given the density and assuming a reuse-1 spectrum allocation within each tier, such networks are interference limited. Using this model, Jo et al.~derived the outage probability and the ergodic rate~\cite{shin:J12} of a multi-tier network with flexible tier association. The authors use the signal-to-interference ratio (SIR) and a pre-specified threshold to define outage at a reference user independent of the load and availability of resources at the serving AP. Estimating the distribution of number of users per AP, this analysis was extended to derive the downlink rate distribution in a multi-tier network~\cite{singh:J13} and a two-tier network with limited backhaul capacity~\cite{singh:J14}.

One major issue in modern heterogeneous networks is tier association. Since the transmit powers of APs in each tier differ, a max-received signal power criterion would result in users connecting to the tier with higher transmit power. Although desirable from the received power point of view, the achieved rate decreases since each user receives a smaller fraction of the resources. Therefore, a popular method is to add, in dB, a chosen bias to the average received power and use this biased received power as the tier association metric. The chosen bias controls the tier association probability. Another issue of importance is dealing with inter- and intra-tier interference. Two approaches have been proposed: 1) spectrum sharing among tiers but with different reuse factors within each~\cite{yicheng:13} or fractional frequency reuse~\cite{novlan:J12}; 2) orthogonal spectrum allocation among tiers, thereby eliminating inter-tier interference (spectrum partitioning, e.g.,\cite{yicheng:14,bao:J14}). 

In this paper, we consider the latter case of spectrum partitioning. We maximize the rate coverage (the complementary cumulative distribution function (CDF) of the achieved rate) by optimizing the tier association probability and the fraction of spectrum allocated to each tier. Optimizing the rate coverage leads to a non-convex  problem, and we obtain the derivatives with respect to each optimization variable to enable efficient numerical solutions. We then explore the intuitive case of setting the tier's allocated fraction of spectrum equal to the tier association probability. In this case, we obtain a closed-form solution for the desired variables. Our numerical results show that the resulting loss in performance is negligible.

\section{System Model}
\label{system}
We consider a network comprising $K$ tiers in the downlink. Within each tier, indexed by $k$, the APs follow a homogeneous PPP $\Phi_{k}$ characterized by $\{P_{k},\lambda_{k}, R_{k}\}$ respectively denoting the transmit power of each AP in the tier, the AP density and the desired rate threshold. The tiers are organized in increasing order of density i.e., $\lambda_{1} \leq \lambda_{2} \cdots \leq \lambda_{K}$. Users are distributed in the network according to an independent PPP with density $\lambda_{u}$. Given the density $\lambda_{k}$, the number of APs belonging to tier $k$ in area $A$ is a Poisson random variable with mean $\lambda_{k}A$, that is independent of other tiers. The channel between an AP and a user is modeled as Rayleigh with average power set by the path loss exponent,~$\alpha$.

The user is associated with the tier with the largest ``biased average power", i.e., the average received power from the potential serving AP in tier $k$ is multiplied by a bias factor $B_{k} \geq 1$; the user then associates with the tier with the highest resulting product. If all the tiers have the same bias factor ($B_{j} = 1, \; \forall j$), the tier association metric is the maximum received power, hence, maximum SIR criterion (here, referred to as max-SIR). In calculating the achievable rate at the user, we assume orthogonal spectrum allocation across tiers with reuse-1 within a tier. Of the total bandwidth $W$, tier $k$ is allocated a fraction $w_k \leq 1$. Therefore, for a user connected to a specific AP in tier $k$, all other APs in tier $k$, but tier $k$ only, act as interferers.

While the tier association is only a function of tier AP density, transmit power and bias factor~\cite{shin:J12}, the actual achieved rate is a function of the allocated spectrum and the load (in terms of the connected users) at each AP. To calculate the reference user's share of spectrum, we use the \emph{average} number of users per AP in tier $k$ given by~\cite{shin:J12}:
\begin{equation}
\label{UperAP}
N_{k} = \frac{A_{k}\lambda_{u}}{\lambda_{k}},
\end{equation}
where $A_{k}$ denotes the $k$th tier association probability. Another expression for the average load, considering AP load as a random variable, is given in~\cite{singh:J13,singh:J14} as $\tilde{N_{k}} = 1 + 1.28 \frac{A_{k}\lambda_{u}}{\lambda_{k}}$; this accounts for the reference user and the implicit area biasing. This higher load will not affect the procedure to derive the overall rate coverage and formulation of the optimization problem in Section~\ref{problem}. For mathematical tractability, however, we use~\eqref{UperAP} to calculate the bandwidth to be allocated to each user; a comparison with the higher average load will be presented through numerical simulations. As in~\cite{singh:J13,yicheng:13,yicheng:14,bao:J14}, we use the proportionally fair model where each AP equally divides its available bandwidth amongst its users. The rate achieved by a user associated with tier $k$ is then given by:
\begin{equation}
\label{rate}
r_{k} = W\frac{w_{k}}{N_{k}}\log_{2}\left(1 + \gamma_{k} \right),
\end{equation}
where $w_{k}$ denotes the fraction of spectrum per AP in tier $k$, and $\gamma_{k}$ denotes the received SIR at the user associated with tier $k$. The user is said to be in coverage if it achieves a threshold rate $R_k$, i.e., $r_k \geq R_k$. In general, $R_k$ is a function of the tier.

\vspace{-0.05in}
\section{Problem Formulation}
\label{problem}
Given a total spectrum of $W$, we aim to optimize the tier association probability, and the spectrum partitioning among tiers to maximize the overall rate coverage. Using~\eqref{rate}, the probability that the user associated with tier $k$ at connection distance $r$ receives its rate threshold is given by:
\begin{equation}
\begin{array}{ll}
\mathds{P}(r_{k} \geq R_{k}\mid r)
& = \mathds{P}(W\frac{w_{k}}{N_{k}}\log_{2}(1 + \gamma_{k}) \geq R_{k}\mid r) \\
& = \mathds{P}(\gamma_{k} \geq 2^{\frac{R_{k}N_{k}}{W w_{k}}} - 1\mid r) \\
&  \qeq \displaystyle e^{-\pi \lambda_{k} r^{2}\rho(\tau_{k}, \alpha)}.
\end{array}
\end{equation}
where $\tau_{k}=2^{\frac{R_{k}N_{k}}{W w_{k}}} - 1$ is the SIR threshold given the rate threshold $R_{k}$, and ($a$) results from the probability of SIR coverage at connection distance $r$ with $ \rho(\tau_{k}, \alpha) = \tau_{k}^{2/\alpha}\int_{\tau_{k}^{-2/ \alpha}}^{\infty}\frac{1}{1 + u^{\alpha/2}}\mathrm{d}u$~\cite[Theorem 2]{andrews:11}. 
 
Having characterized the rate coverage in a single tier with average load per AP, the probability that the user is in coverage in a multi-tier network is given by the following Theorem.
\begin{thm}
\label{RateCoverage}
In a $K$-tier network with orthogonal spectrum allocation across tiers, and APs in each tier distributed according to a homogeneous PPP with density $\lambda_{k}$, the probability of the rate coverage is given by:
\begin{equation}
\boldsymbol{R_{c}} = \sum_{k = 1}^{K}\frac{1}{A_{k}^{-1} + \rho(\tau_{k}, \alpha)},
\end{equation}
where $\tau_{k} = 2^{\frac{R_{k}N_{k}}{Ww_{k}}} - 1$, and $A_{k}$ denotes the association probability to tier $k$.
\end{thm}
\begin{proof}
See the Appendix.
\end{proof}
\vspace{-0.1in}
\subsection{Optimization Problem}
Using the result derived in Theorem~\ref{RateCoverage}, the optimization problem with the objective of maximizing the total probability of rate coverage can be formulated as:
\vspace{-0.05in}
\begin{equation}
\begin{array} {ll}
\label{OptProb}
\displaystyle \max_{\{A_{k}\}_{k =1}^{K},\{w_{k}\}_{k =1}^{K}} & \displaystyle \sum_{k = 1}^{K} \frac{1}{A_{k}^{-1} + \rho(\tau_{k}, \alpha)} \\
\mbox{subject to:}  & \displaystyle \sum_{k = 1}^{K} A_{k} = 1, \sum_{k = 1}^{K} w_{k} = 1 \\
& A_{k} \geq 0, w_{k} \geq 0 \; \; k = 1,\cdots K. \\
\end{array}
\end{equation}
Although this optimization problem is non-convex, there is a relation between the first derivative of the objective function with respect to each pair of optimization variables which simplifies the gradient-based schemes to obtain the local optima. Defining $f_{k} (A_{k},w_{k})= \frac{A_{k}}{1 + A_{k}\rho(\tau_{k}, \alpha)}$, the equivalent unconstrained objective function is given by:
\vspace{-0.06in}
\begin{eqnarray}
\label{LagrangeOrig}
\mathcal{L}(A_{k}, \eta, \mu) & = &  \sum_{k = 1}^{K} f_{k}(A_{k}, w_{k}) - \eta \left(\sum_{k = 1}^{K}A_{k}- 1\right) \nonumber \\ & & - \mu\left(\sum_{k=1}^{K}w_{k} - 1 \right),
\end{eqnarray}
where $\eta$ and $\mu$ are the Lagrangian multipliers. The KKT conditions (in addition to two equality constraints) are:
\vspace{-0.06in}
\begin{equation}
\label{KKT}
\frac{\partial f_{k}(A_{k}, w_{k})}{\partial A_{k}} = \eta, \; \; \; \frac{\partial f_{k}(A_{k}, w_{k})}{\partial w_{k}} = \mu \; \; \; \forall k.
\end{equation}
The derivative of $f_{k}(A_{k}, w_{k})$ w.r.t.~$w_{k}$ is given by:
\begin{equation}
\label{FtoW}
\frac{\partial f_{k}(A_{k}, w_{k})}{\partial w_{k}} = \frac{- \frac{\partial \rho(\tau_{k},\alpha)}{\partial w_{k}}}{(A_{k}^{-1} + \rho(\tau_{k}, \alpha))^2},
\end{equation}
\begin{equation}
\label{eq1}
\hspace*{-1.1in} \mathrm{where} \hspace*{0.3in} \frac{\partial \rho(\tau_{k},\alpha) }{\partial w_{k}} =  \frac{\partial \rho(\tau_{k},\alpha)}{\partial \tau_{k}}\cdot \frac{\partial \tau_{k}}{\partial w_{k}},
\end{equation}
\begin{equation}
\frac{\partial \rho(\tau_{k},\alpha)}{\partial \tau_{k}} = \frac{2}{\alpha \tau_{k}}\left[ \rho(\tau_{k}, \alpha) + \frac{1}{1 + \tau_{k}^{-1}}\right],
\end{equation}
\begin{equation}
\label{eq2}
\frac{\partial \tau_{k}}{\partial w_{k}} = \log(2) (\frac{-R_{k}N_{k}}{W w_{k}^{2}})2^{R_{k}N_{k}/W w_{k}}.
\end{equation}
Using~\eqref{eq1}-\eqref{eq2} in~\eqref{FtoW}, and $2^{R_{k}N_{k}/Ww_{k}} = \tau_{k} + 1$, we have:
\begin{multline}
\label{FinalToW}
\frac{\partial f_{k}(A_{k}, w_{k})}{\partial w_{k}}  = 
\\  \frac{2 \log(2) / \alpha}{(A_{k}^{-1} + \rho(\tau_{k}, \alpha))^2}\left[\frac{(1+\tau_{k}) \rho(\tau_{k}, \alpha)+ \tau_{k}}{\tau_{k}} \left(\frac{R_{k}N_{k}}{W w_{k}^{2}} \right) \right]. 
\end{multline}
Similarly, the derivative of $f_{k}(A_{k}, w_{k})$ w.r.t.~$A_{k}$ is given by:
\begin{equation}
\label{eq4}
\frac{\partial f_{k}(A_{k}, w_{k})}{\partial A_{k}} = \frac{1 - A_{k}^{2}\partial \rho(\tau_{k},\alpha)/ \partial A_{k}}{(1 + A_{k}\rho(\tau_{k}, \alpha))^{2}},
\end{equation}
\begin{equation}
\hspace*{-1.1in} \mathrm{where} \hspace*{0.3in}\frac{\partial \rho(\tau_{k},\alpha)}{\partial A_{k}}  =  \frac{\partial \rho(\tau_{k},\alpha)}{\partial \tau_{k}}\cdot \frac{\partial \tau_{k}}{\partial A_{k}},
\end{equation}
\begin{equation}
\frac{\partial \tau_{k}}{\partial A_{k}} = \log(2) (\frac{R_{k}\lambda_{u}}{\lambda_{k}W w_{k}})2^{R_{k}N_{k}/W w_{k}}.
\end{equation}
Hence,
\begin{equation}
\label{eq5}
\frac{\partial \rho(\tau_{k}, \alpha)}{\partial A_{k}} = \frac{2 \log(2) }{\alpha}\left[ \frac{(1 + \tau_{k})\rho(\tau_{k}, \alpha) + \tau_{k}}{\tau_{k}}\left(\frac{R_{k}\lambda_{u}}{\lambda_{k}Ww_{k}}\right) \right].
\end{equation}
Using~\eqref{eq5} in~\eqref{eq4} results in:
\begin{equation}
\label{FinalToA}
\frac{\partial f_{k}(A_{k}, w_{k})}{\partial A_{k}} = \frac{1 - A_{k}^{2}\frac{2 \log(2)}{ \alpha }\left[\frac{(1 + \tau_{k})\rho(\tau_{k}, \alpha)  + \tau_{k}}{\tau_{k}} (\frac{R_{k}\lambda_{u}}{\lambda_{k}Ww_{k}})\right] }{(1 + A_{k}\rho(\tau_{k}, \alpha))^{2}}.
\end{equation}
Comparing~\eqref{FinalToW} with~\eqref{FinalToA}, we have:
\begin{equation}
\frac{\partial f_{k}(A_{k}, w_{k})}{\partial A_{k}} = \frac{1}{(1 + A_{k}\rho(\tau_{k}, \alpha))^{2}} - \frac{w_{k}}{A_{k}}\frac{\partial f_{k}(A_{k}, w_{k})}{\partial w_{k}}.
\label{eq:relation}
\end{equation}
Note that the conditions in~\eqref{KKT} imply that any solution to the optimization problem must satisfy the fixed point equation:
\begin{equation}
\eta = \frac{1}{(1 + A_{k}\rho(\tau_{k}, \alpha))^{2}} - \frac{w_{k}}{A_{k}} \mu.
\end{equation}
However, the complicated relation between $\rho(\tau_{k}, \alpha)$ and the variables $A_{k}$ and $w_{k}$ makes it difficult to derive a closed-form solution. Hence, we will use~\eqref{FinalToW} and~\eqref{eq:relation} to simplify the interior-point method to solve the optimization problem.

\subsection{Equating the Two Fractions}
This section explores a simple, and intuitive, solution to the optimization problem in~\eqref{OptProb}. The optimization variables are the fraction of users associated with each tier ($A_k$) and the fraction of spectrum allocated to that tier ($w_k$). Therefore, an intuitive solution is to set $A_k = w_k$, i.e., to use the same fraction for both variables. As our results will show, this solution is extremely close to the (numerical) solution to~\eqref{OptProb}.
In the case of $A_k = w_k$, the optimization problem is reduced to:
\begin{equation}
\begin{array} {ll}
\label{OptProb2}
\displaystyle \max_{\{A_{k}\}_{k =1}^{K}} & \displaystyle \sum_{k = 1}^{K} \frac{1}{A_{k}^{-1} + \rho(\bar{\tau_{k}}, \alpha)} \\
\mbox{subject to:}  & \displaystyle \sum_{k = 1}^{K} A_{k} = 1 \\
& A_{k} \geq 0 \; \; k = 1,\cdots K,
\end{array}
\end{equation}
where $\{A_{k}\}_{k = 1}^{K}$ is now the only set of optimization variables and $\bar{\tau_{k}} = 2^{R_{k}\lambda_{u}/W\lambda_{k}} - 1$. It is easy to show that this problem is concave. The equivalent unconstrained objective function is:
\begin{equation}
\label{Lagrange}
\mathcal{L}(A_{k}, \eta) = \sum_{k = 1}^{K} \frac{A_{k}}{1 + A_{k}\rho(\bar{\tau_{k}}, \alpha)} - \eta \left(\sum_{k = 1}^{K}A_{k}- 1\right).
\end{equation}
Differentiating~\eqref{Lagrange} with respect to $A_{k}$, and setting the derivative to $0$, we obtain:
\begin{equation}
\frac{\partial\mathcal{L}(A_{k}, \eta)}{\partial A_{k}} = \frac{1}{(1 +A_{k}\rho(\bar{\tau_{k}}, \alpha))^{2}} - \eta = 0,
\end{equation}
\begin{equation}
\label{AOpt}
\Longrightarrow A_{k} = \frac{\sqrt{1/\eta} - 1}{\rho(\bar{\tau_{k}}, \alpha)}.
\end{equation}
Applying $\sum_{k=1}^{K}A_{k} = 1$, we have $\sqrt{1/ \eta} - 1 = \displaystyle \frac{1}{\sum_{k=1}^{K}1/\rho(\bar{\tau_{k}}, \alpha)}$. Using this expression in~\eqref{AOpt}, the optimum tier association and spectrum allocation for tier $k$ is given by:
\begin{equation}
\label{OptValues}
A_{k}^{*} = w_{k}^{*} = \displaystyle \frac{1/ \rho(\bar{\tau_{k}}, \alpha)}{\sum_{k = 1}^{K}1/\rho(\bar{\tau_{k}},\alpha)},
\end{equation}
i.e., in this special case when we equate the two fractions, we have a closed-form solution to~\eqref{OptProb}.
\section{Simulation Results}
\label{simulations}
We consider a three-tier network ($K=3$) with $\lambda_{u} = 5/100, \{\lambda_{1}, \lambda_{2}, \lambda_{3}\} = \{0.01, 0.05, 0.2\}\lambda_{u}$ and $\{P_{1},P_{2}, P_{3}\}$ = \{46, 30, 20\} dBm denoting the user density, tiers' AP density and transmit power respectively. We obtain the optimum tier association probability and spectrum partitioning for three different scenarios: 1) $\{A_{k}\}_{k = 1}^{K}$ and $\{w_{k}\}_{k = 1}^{K}$ are interior-point solutions to the optimization problem in~\eqref{OptProb}; 2) $\{A_{k} = w_{k}\}_{k = 1}^{K}$ are solutions using~\eqref{OptValues}; 3) $\{w_{k}\}_{k = 1}^{K}$ are solutions to the optimization problem in~\eqref{OptProb} when $B_{k} = 1 \;\forall k$, i.e., the max-SIR scenario. We compare the obtained results with the optimum solution through a brute force search. The optimum tier association and spectrum partitioning with the higher average load per AP, $\tilde{N_{k}}$, are also presented for comparison. We use $\alpha = 3.5$ as the path loss exponent for all tiers.

Fig.~\ref{fig:RateCov} shows the overall rate coverage for equal and different tier rate thresholds. Clearly, the max-SIR performs much worse than optimizing the relevant fractions, illustrating the advantage of offloading (if done jointly with the resource allocation). More interesting is the rate coverage achieved when the tier's share of spectrum is equal to the share of users it serves as given by~\eqref{OptValues}. While the overall network coverage is almost identical to the optimum case, there is a slight difference in tier association and spectrum partitioning as shown in Fig.~\ref{fig:TierAsso}. Note that if the rate threshold increases, the tier's probability of coverage decreases. Hence, the network coverage is maximized by moving users (followed by the required spectrum) from the tier with the increasing rate threshold to the other tiers.
  
\begin{figure}
     \begin{center}
        \subfigure[Same rate threshold for all tiers.]{%
            \label{fig:ER}
            \includegraphics[width=0.39\textwidth]{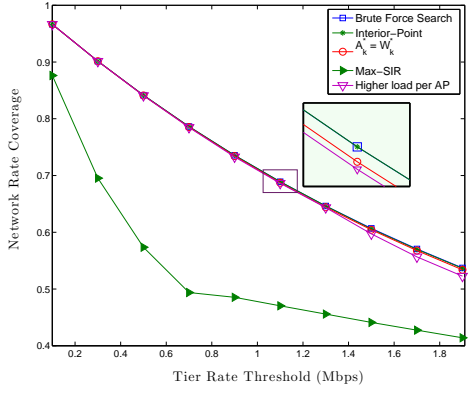}
        }%
        \\
        \subfigure[Different rate thresholds; $\{R_{1},R_{2}\} = \{0.5,1\}$ Mbps.]{%
           \label{fig:DR}
           \includegraphics[width=0.39\textwidth]{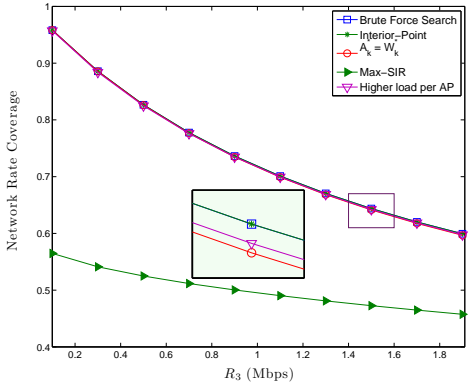}
        } 
        \end{center}
    \caption{Overall rate coverage in a 3-tier network with (a) the same rate threshold for all tiers, (b) different rate thresholds. In both cases, $\{P_{1},P_{2},P_{3}\}$ = \{46, 30, 20\}dBm and $\{\lambda_{1},\lambda_{2},\lambda_{3}\} = \{0.01,0.05,0.2\} \times \lambda_{u}$.}
   \label{fig:RateCov}
\end{figure}

\begin{figure}
     \begin{center}
        		\includegraphics[width=0.5\textwidth]{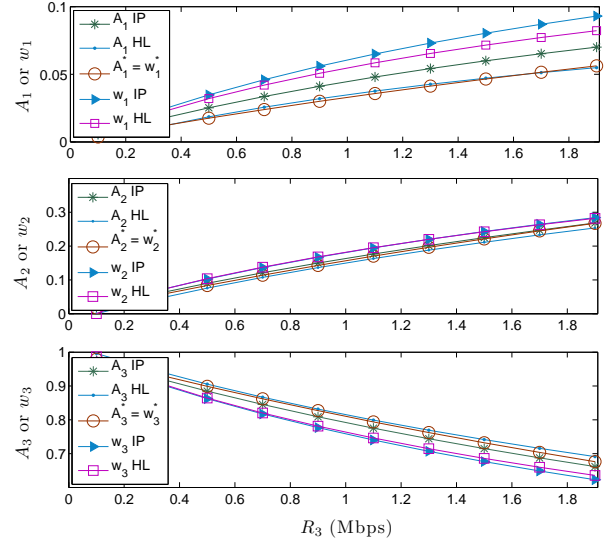}
        	\end{center}
	
    \caption{Comparing the optimum tier association and spectrum partitioning for different tiers with the solution to \eqref{OptProb2}, i.e, $A_{k}^{*} = w_{k}^{*}$. $\{P_{1},P_{2},P_{3}\} = \{46,30,20\}$dBm and $\{\lambda_{1},\lambda_{2},\lambda_{3}\} = \{0.01,0.05,0.2\} \times \lambda_{u}$. The results obtained by the Interior-point method and the optimum values considering the higher average load $\tilde{N_{k}}$ are referred to as `IP' and `HL' respectively.}
          \label{fig:TierAsso}
\end{figure}

\section{Conclusions}
We considered the problem of optimizing the tier association probability and spectrum portioning in a multi-tier network with the objective of maximizing the rate coverage. Our results show a significantly improved coverage by jointly optimizing the user association and spectrum allocation. Interestingly, the intuitive solution of equating the two fractions results in negligible performance loss. This result is important from the system design point of view: (i) it simplifies the optimization problem reducing it to one with a closed-form solution given by~\eqref{OptValues}; (ii) the tier with the smallest fraction of spectrum also serves the least number of users. Considering a reasonable threshold for $A_{k}$ (hence $w_{k}$), a tier can potentially be eliminated from the network with little impact on the network rate coverage.
\label{conclusion}

\appendix[Proof of Theorem \ref{RateCoverage}]
\label{AppendixA}
\begin{proof}
This is a special case of the rate coverage derived in~\cite{singh:J13} with average number of users per AP and orthogonal spectrum allocation across tiers. The probability that a user connects to tier $k$ at connection distance $r$ is given by~\cite[Lemma 1]{shin:J12}~$\mathds{P}(n = k\mid r)=\prod_{j = 1,j\neq k}^{K}e^{-\pi \lambda_{j}(P_{j}B_{j}/P_{k}B_{k})^{2/\alpha}r^{2}}.$ Therefore, the probability of the joint event that the user connects to tier $k$ and meets its rate threshold is given by:
\begin{equation}
\begin{array}{l}
\mathds{P}(r_{k} \geq R_{k}, n = k)  = \mathds{E}_{r} \Big[\mathds{P}(\gamma_{k} \geq \tau_{k}, n = k\mid r)\Big] \\
 = \mathds{E}_{r} \Big[\mathds{P}(\gamma_{k} \geq \tau_{k}\mid r)\cdot \mathds{P}(n = k\mid r)\Big] \\
= \displaystyle \int_{r=0}^{\infty} e^{-\pi \lambda_{k} r^{2}\rho(\tau_{k}, \alpha)} \cdot \left(\prod_{j = 1,j\neq k}^{K}e^{-\pi \lambda_{j}\left(\frac{P_{j}B_{j}}{P_{k}B_{k}}\right)^{2/\alpha}r^{2}}\right)f_{r}\left(r\right)\mathrm{d}r \\
\qeq \displaystyle \int_{r=0}^{\infty} 2\pi\lambda_{k}r e^{-\pi \lambda_{k} r^{2} \Big[  \rho(\tau_{k}, \alpha) + \sum_{j=1}^{K} \frac{\lambda_{j}}{\lambda_{k}}\left(\frac{P_{j}B_{j}}{P_{k}B_{k}}\right)^{2/\alpha}\Big]} \mathrm{d}r \\
= \frac{1}{\left(A_{k}^{-1} + \rho(\tau_{k}, \alpha)\right)},
\end{array}
\end{equation}
where ($a$) results from the distribution of the connection distance in a PPP network with density $\lambda_{k}$ given by $f_{r}\left(r\right) = 2\pi \lambda_{k}r e^{-\pi \lambda_{k} r^{2}}$, and $A_{k}^{-1}=\sum_{j=1}^{K} \frac{\lambda_{j}}{\lambda_{k}}\left(\frac{P_{j}B_{j}}{P_{k}B_{k}}\right)^{2/\alpha}$~\cite{shin:J12}. Note that we do not consider a random load at each AP, but constant average load only affected by the user and AP densities and the tier association probabilities. Using the sum probability of disjoint events, the overall probability of rate coverage is:
\begin{equation}
\boldsymbol{R_{c}} = \displaystyle \sum_{k=1}^{K}\mathds{P}(r_{k} \geq R_{k}, n = k) = \displaystyle \sum_{k=1}^{K}\frac{1}{A_{k}^{-1} + \rho(\tau_{k}, \alpha)},
\end{equation}
and the proof is complete.
\end{proof}

\bibliographystyle{IEEEtran.bst}
\bibliography{RefCommLett}
\end{document}